\documentclass[letterpaper,10pt,conference]{ieeeconf}

\IEEEoverridecommandlockouts
\let\labelindent\relax

\usepackage{xcolor} 
\usepackage{amsmath,amssymb,amsfonts,amsthm,mathtools}
\usepackage{cite}
\usepackage{graphicx}
\usepackage{booktabs}
\usepackage{enumitem}
\usepackage{algorithm}
\usepackage{algpseudocode}
\usepackage{threeparttable}

\algrenewcommand\algorithmicrequire{\textbf{Input:}}
\algrenewcommand\algorithmicensure{\textbf{Output:}}

\newtheorem{assumption}{Assumption}
\newtheorem{lemma}{Lemma}
\newtheorem{proposition}{Proposition}
\newtheorem{theorem}{Theorem}
\newtheorem{corollary}{Corollary}
\newtheorem{remark}{Remark}

\newcommand{\vect}[1]{\boldsymbol{#1}}
\newcommand{\R}{\mathbb{R}}
\newcommand{\Rpos}{\mathbb{R}_{\geq 0}}
\newcommand{\cW}{\mathcal{W}}
\newcommand{\cX}{\mathcal{X}}
\newcommand{\cU}{\mathcal{U}}
\newcommand{\cD}{\mathcal{D}}
\newcommand{\cB}{\mathcal{B}}
\newcommand{\rfun}{\vect{r}}
\newcommand{\range}{\operatorname{range}}

\title{Anytime Primal--Dual Certification of the Maximum\\Disturbance Radius in Robust MPC}

\author{Wenqi Cai$^{1}$,
Muhammad Bakr Abdelghany$^{2}$,
Kyriakos G. Vamvoudakis$^{3}$,
and Anthony Tzes$^{1}$
\thanks{$^{1}$Wenqi Cai and Anthony Tzes are with the Electrical Engineering Program, New York University Abu Dhabi (NYUAD), 129188, United Arab Emirates. Email: wenqi.cai@nyu.edu; anthony.tzes@nyu.edu}
\thanks{$^{2}$Muhammad B. Abdelghany is with the Electrical Engineering Department, Khalifa University, United Arab Emirates (e-mail: muhammad.babdelghany@ku.ac.ae).}
\thanks{$^{3}$Kyriakos G. Vamvoudakis is with the School of Aerospace Engineering, Georgia Institute of Technology, Atlanta, GA 30332, United States. Email: kyriakos@gatech.edu}
\thanks{This work was supported in part by: a) NSF under grant Nos. SLES-$2415479$, CPS-$2227185$, b) NASA ULI under grant No. $80$NSSC$25$M$7104$, and c) the NYUAD Center for AI and Robotics, funded by Tamkeen under the NYUAD Research Institute Award CG010.}
}

\begin{document}
\maketitle
\thispagestyle{empty}
\pagestyle{empty}

\begin{abstract}
Adjustable-set robust model predictive control (MPC) characterizes a state-dependent maximum disturbance radius, which can be interpreted as a certified robustness reserve. Existing methods primarily focus on optimizing and propagating this reserve under in-set disturbances. This letter investigates how much reserve remains after a finite out-of-set disturbance without immediately re-solving the full optimization problem. To this end, we propose a two-sided reserve-depletion envelope formed by independent primal and dual correction hierarchies, which supply monotone lower and upper bounds, respectively. The envelope remains valid across active-set changes, has an online-computable width, and contracts monotonically under subspace expansion. Leveraging its finite-step exactness, we present a basis-first adaptive algorithm that operates in an anytime manner: every completed reduced solve returns a valid certificate, enabling early termination once a prescribed tolerance is reached. Across the tested cases, numerical studies report no certificate violations and median speedups of up to $4.97\times$ over warm-started full re-optimization at a $2\%$ certificate-width tolerance.
\end{abstract}

\section{Introduction}
\par Geometric margins alone do not characterize disturbance tolerance in constrained dynamical systems. States with the same distance to a constraint boundary can have markedly different recoverability because dynamics, input limits, and future constraint interactions govern robust feasibility. In the adjustable-set tube-based robust MPC formulation considered here, this tolerance is quantified by the state-dependent maximum disturbance radius: the largest scaling of a prescribed disturbance set for which the finite-horizon problem under fixed ancillary feedback remains feasible. We refer to this scalar as the \emph{certified robustness reserve}, or simply the \emph{reserve}. When a realized disturbance exceeds the certified set, the realized successor is evaluated against the modeled successor obtained by projecting the disturbance onto that set. Their signed reserve difference may represent either depletion or gain. Exact evaluation of this \emph{signed reserve depletion} requires solving the full reserve linear program (LP) at both successor states, which can be burdensome under limited or variable online computation.

\par Robust MPC provides several mechanisms for constrained operation under bounded disturbances. Classical robust-feasibility conditions characterize when admissible operation can be sustained \cite{kerrigan2001robust}, while tube-based MPC enforces recursive constraint satisfaction through invariant error sets \cite{mayne2005robust}. Adjustable uncertainty-set formulations jointly optimize the admissible set and control policy \cite{zhang2017robust,kim2018robust}, with zonotopic constructions yielding tractable adjustable tubes \cite{raghuraman2021tube}. Related approaches use reactive safety modes \cite{carson2013safety}, feasibility governors \cite{skibik2021feasibility}, closed-loop re-solvability conditions \cite{parsi2024once}, and reference or command governors \cite{garone2017reference}. These methods provide control-level mechanisms for maintaining or recovering admissible operation, yet they do not directly quantify the residual robustness of the system after perturbation.

\par Assessing this change online introduces a distinct computational challenge. Multiparametric programming provides explicit solutions over critical regions \cite{tondel2003algorithm}, whereas LP post-optimality analysis characterizes local primal and dual variations while the relevant basis or partition is preserved \cite{greenberg2000sensitivity}. Warm-started active-set methods accelerate repeated solves \cite{ferreau2008online,xie2023maximal}, while real-time iteration, suboptimal MPC, and time-distributed optimization limit per-sample computation \cite{diehl2005real,zeilinger2011realtime,liao2020time}. Robust-to-early-termination MPC, primal--dual stopping criteria, and learned primal--dual policies retain selected closed-loop or optimization guarantees under incomplete computation \cite{hosseinzadeh2023robust,pavlov2019early,zhang2019safe}. Collectively, these developments provide tools for characterizing parametric solutions, accelerating repeated optimization, and operating with incomplete solves. The complementary question considered here is how the robustness reserve changes relative to the modeled successor after a finite out-of-set perturbation, and how this signed change can be certified under a limited online computational budget.

\par To address this question, this work develops an anytime primal--dual certification framework for constrained linear systems under fixed tube feedback. Using the shifted feasible plan and the current optimal multiplier as primal and dual anchors, respectively, we construct two independently expandable certificate hierarchies. The main contributions are:
\begin{enumerate}
[leftmargin=*,itemsep=1pt,topsep=2pt]
\item \emph{Signed Primal Hierarchy:} We construct nested primal-feasible lower reserve certificates with monotone improvement, an implementable primal repair, signed reserve adjustment capturing both depletion and gain, and finite exactness;
\item \emph{Reduced Dual Hierarchy:} We construct nested dual nullspace corrections yielding valid upper reserve certificates with monotone tightening and finite exactness;
\item \emph{Anytime Reserve Certification:} We combine both hierarchies into a two-sided reserve-depletion envelope and construct a basis-first adaptive algorithm that maintains validity after every completed reduced solve and across active-set changes, supports early termination under prescribed tolerance, and becomes exact after finitely many expansions.
\end{enumerate}


\emph{Notation:} $\Rpos$ denotes the set of nonnegative reals. For a matrix $M$, $M_j$ denotes its $j$th row, while $\range(M)$ and $\ker(M)$ denote its range and nullspace, respectively. Vector inequalities are entrywise. For a nonempty closed convex set $\mathcal S$, $\Pi_{\mathcal S}(\cdot)$ denotes the Euclidean projection onto $\mathcal S$.
\section{Maximum Disturbance Radius}
\subsection{System and reserve formulation}

\par Consider the constrained linear system
\begin{equation}
\vect{x}_{k+1}=A\vect{x}_{k}+B\vect{u}_{k}+E\vect{w}_{k},
\label{eq:system}
\end{equation}
where $\vect{x}_{k}\in\R^{n_x}$, $\vect{u}_{k}\in\R^{n_u}$, and $\vect{w}_{k}\in\R^{n_w}$. The state and input constraints are polytopic, $\vect{x}_{k}\in\cX\coloneqq\{\vect{x}:F_x\vect{x}\leq\vect{g}_x\}$ and $\vect{u}_{k}\in\cU\coloneqq\{\vect{u}:F_u\vect{u}\leq\vect{g}_u\}$. Let $\cW_0\subset\R^{n_w}$ be a compact polytope containing the origin, and let $\delta\cW_0$ with scalar $\delta\geq0$ denote the modeled disturbance set. 

\par We adopt a standard tube-based robust MPC formulation with fixed ancillary feedback and robust terminal ingredients. The support-function tightenings induced by $\delta\cW_0$ scale linearly with $\delta$. After eliminating the predicted states, let $\vect{z}\in\R^{n_z}$ collect the nominal input sequence and the remaining auxiliary variables. The resulting state-dependent maximum admissible scaling $\rho(\vect{x})$, referred to as the \emph{reserve}, is given by the parametric linear program (LP)
\begin{subequations}
\label{eq:reserveLP}
\begin{align}
\rho(\vect{x})\coloneqq\max_{\vect{z},\delta}&\quad\delta
\label{eq:reserveLP_obj}\\
\mathrm{s.t.}\quad&G\vect{z}+\vect{q}\delta\leq\vect{b}+H\vect{x},
\label{eq:reserveLP_con}\\
&\delta\geq0.
\label{eq:reserveLP_nonneg}
\end{align}
\end{subequations}
Here, $G\in\R^{n_c\times n_z}$, $\vect{q},\vect{b}\in\R^{n_c}$, and $H\in\R^{n_c\times n_x}$, where $n_c$ is the number of scalar inequalities in the condensed formulation. All constraints on $\vect{z}$ are included in \eqref{eq:reserveLP_con}. The vector $\vect{q}\geq\vect{0}$ collects the support-function-based tightening coefficients, with $q_j$ denoting the tightening of the $j$th inequality per unit increase in $\delta$.

\par Define $\rfun(\vect{x})\coloneqq\vect{b}+H\vect{x}$, domain $\cX_{\mathrm f}\coloneqq\{\vect{x}\in\cX:\eqref{eq:reserveLP}\text{ is feasible and bounded}\}$, and domain $\cX_\rho\coloneqq\{\vect{x}\in\cX_{\mathrm f}:\rho(\vect{x})>0\}$. 
\begin{assumption}[Shift feasibility]
\label{ass:shift}
For any feasible pair $(\vect{z}_{k},\delta_k)$ of \eqref{eq:reserveLP} at $\vect{x}_k$ and any modeled disturbance $\vect{w}_{k}^{0}\in\delta_k\cW_0$, applying the first control move $\vect{u}_k$ encoded by $\vect{z}_k$ admits a shifted decision vector $\mathsf{S}(\vect{z}_{k},\vect{w}_{k}^{0})$ satisfying
\begin{equation}
G\mathsf{S}(\vect{z}_{k},\vect{w}_{k}^{0})+\vect{q}\delta_k
\leq\rfun(\vect{x}_{k+1}^{0}),
\label{eq:shift_feasible}
\end{equation}
where $\vect{x}_{k+1}^{0}$ is the successor state under $\vect{w}_{k}^{0}$.
\end{assumption}
Assumption~\ref{ass:shift} is the standard tube-MPC shift condition and can be enforced via standard robust invariant terminal sets and terminal controllers\cite{mayne2005robust,kim2018robust}.
\subsection{Dual characterization}

\par The dual-feasible set associated with \eqref{eq:reserveLP} is
\begin{equation}
\cD\coloneqq\left\{\vect{\lambda}\in\Rpos^{n_c}:G^\top\vect{\lambda}=\vect{0},\ \vect{q}^\top\vect{\lambda}\geq1\right\}.
\label{eq:dual_set}
\end{equation}

\begin{lemma}[Parametric reserve bounds]
\label{lem:duality}
For every $\vect{x}\in\cX_{\mathrm f}$, strong duality holds and
\begin{equation}
\rho(\vect{x})=\min_{\vect{\lambda}\in\cD}\vect{\lambda}^{\top}\rfun(\vect{x}).
\label{eq:dual_problem}
\end{equation}
Consequently, any pair $(\tilde{\vect{z}},\tilde{\delta})$ that is primal-feasible at $\vect{x}$ and any $\tilde{\vect{\lambda}}\in\cD$ satisfy
\begin{equation}
\tilde{\delta}\leq\rho(\vect{x})\leq\tilde{\vect{\lambda}}^{\top}\rfun(\vect{x}).
\label{eq:generic_bounds}
\end{equation}
Moreover, $\rho$ is concave and piecewise affine on $\cX_{\mathrm f}$.
\end{lemma}

\begin{proof}
The Lagrangian of \eqref{eq:reserveLP} is $\delta+\vect{\lambda}^{\top}\bigl(\rfun(\vect{x})-G\vect{z}-\vect{q}\delta\bigr)$. Its supremum over the free vector $\vect{z}$ and $\delta\geq0$ is finite exactly when $\vect{\lambda}\in\cD$. Since \eqref{eq:reserveLP} is feasible and bounded on $\cX_{\mathrm f}$, strong duality yields \eqref{eq:dual_problem}, while \eqref{eq:generic_bounds} follows from weak duality. The final claim follows from the polyhedral parametric-LP representation in \eqref{eq:dual_problem}.
\end{proof}

\par For any disturbance-induced state displacement $\vect{d}\in\R^{n_x}$ such that $\vect{x}+\vect{d}\in\cX_{\mathrm f}$,
\begin{equation} 
\rho(\vect{x}+\vect{d}) \leq(\vect{\lambda}^{\star})^{\top}\rfun(\vect{x}+\vect{d}) =\rho(\vect{x})+(\vect{\lambda}^{\star})^{\top}H\vect{d}. \label{eq:dual_supergradient} 
\end{equation} 
Indeed, $\cD$ is state-independent, so the optimal multiplier $\vect{\lambda}^{\star}$ remains dual-feasible at $\vect{x}+\vect{d}$ and provides an immediate upper bound, whereas a certified lower bound requires a primal-feasible recourse. This asymmetry motivates the primal and dual correction hierarchies developed next.

\section{Primal--Dual Reserve Certificates}
\par At time $k$, let $(\vect{z}_{k}^{\star},\rho_k)$ and $\vect{\lambda}_{k}^{\star}$ be primal and dual optimizers of \eqref{eq:reserveLP} at $\vect{x}_k\in\cX_\rho$, with $\rho_k\coloneqq\rho(\vect{x}_k)$. Upon applying the first control input encoded by $\vect{z}_{k}^{\star}$, decompose the realized disturbance as
\begin{equation}
\vect{w}_{k}=\vect{w}_{k}^{0}+\vect{e}_{k},\qquad
\vect{w}_{k}^{0}\coloneqq\Pi_{\rho_k\cW_0}(\vect{w}_k),
\label{eq:dist_decomp}
\end{equation}
where $\vect{e}_{k}\coloneqq\vect{w}_{k}-\vect{w}_{k}^{0}$ is the excess perturbation relative to the certified set $\rho_k\cW_0$. Define the modeled and realized (excess-perturbed) successors, indexed by $0$ and $e$, as
\begin{align}
\vect{x}_{k+1}^{0}
&=A\vect{x}_{k}+B\vect{u}_{k}+E\vect{w}_{k}^{0},
\label{eq:x0}\\
\vect{x}_{k+1}^{e}
&=\vect{x}_{k+1}^{0}+\vect{d}_{k},\qquad
\vect{d}_{k}\coloneqq E\vect{e}_{k}.
\label{eq:xe}
\end{align}

\par By Assumption~\ref{ass:shift}, the shifted plan $\vect{z}_{k+1}^{0}\coloneqq\mathsf{S}(\vect{z}_{k}^{\star},\vect{w}_{k}^{0})$ is primal-feasible at $\vect{x}_{k+1}^{0}$ with reserve level $\rho_k$. Its constraint slack is
\begin{equation}
\vect{s}_{k+1}^{0}
\coloneqq
\rfun(\vect{x}_{k+1}^{0})
-G\vect{z}_{k+1}^{0}
-\vect{q}\rho_k
\geq\vect{0}.
\label{eq:slack}
\end{equation}
Then, for $\sigma\in\{0,e\}$, set $\vect{d}_{k}^{0}\coloneqq\vect{0}$ and $\vect{d}_{k}^{e}\coloneqq\vect{d}_{k}$, and define
\begin{equation}
\vect{x}_{k+1}^{\sigma}\coloneqq\vect{x}_{k+1}^{0}+\vect{d}_{k}^{\sigma},\quad
\bar{\vect{z}}_{k+1}^{\sigma}\coloneqq\vect{z}_{k+1}^{0}+L\vect{d}_{k}^{\sigma}.
\label{eq:successor_notation}
\end{equation}
Here, $L\in\R^{n_z\times n_x}$ is a fixed linear prediction map generating a baseline feedforward update, and $\bar{\vect{z}}_{k+1}^{\sigma}$ is the resulting uncorrected primal anchor. In what follows, $\underline{\rho}$ and $\overline{\rho}$ denote primal lower and dual upper reserve certificates, respectively.
\subsection{Primal correction hierarchy}

\par To construct the default primal hierarchy, let $P_u\in\R^{n_z\times Nn_u}$ embed stacked nominal-input corrections into $\vect{z}$, and select $K_r$ such that $A_r\coloneqq A+BK_r$ is Schur. This yields
\begin{equation}
\mathcal K_N
\coloneqq
\operatorname{col}(K_r,K_rA_r,\ldots,K_rA_r^{N-1}),
\,\,\,
L\coloneqq P_u\mathcal K_N.
\label{eq:Lcl}
\end{equation}
Let $T_m$ span the level-$m$ primal correction subspace. For $m=0,\ldots,N$, define 
\begin{equation}
J_m\coloneqq
\begin{bmatrix}
I_{mn_u}\\
0_{(N-m)n_u\times mn_u}
\end{bmatrix},
\quad
T_m\coloneqq P_uJ_m,
\label{eq:Tm}
\end{equation}
with $\range(T_0)=\{\vect{0}\}$, i.e., $m=0$ admits no correction. This prefix construction \eqref{eq:Tm} progressively frees corrections over the first $m$ inputs. Finally, append the full-space endpoint $T_{M_{\mathrm p}}\coloneqq I_{n_z}$ with $M_{\mathrm p}=N+1$. Thus,
\begin{equation}
\range(T_0)\subseteq\cdots\subseteq
\range(T_{M_{\mathrm p}})=\R^{n_z}.
\label{eq:primal_nesting}
\end{equation}
Note that while the choices of $L$ and $T_m$ affect reduced-level performance, the analysis below requires only the nesting and full-space endpoint in \eqref{eq:primal_nesting}.

\par Restricting the correction to $\range(T_m)$ at level $m$, define the signed reserve-adjustment LP
\begin{subequations}
\label{eq:signed_primal_LP}
\begin{align}
\beta_{k,m}^{\sigma\star}
\coloneqq
\min_{\vect{\eta},\beta}
&\quad \beta
\label{eq:signed_primal_obj}\\
\mathrm{s.t.}\quad
&GT_m\vect{\eta}-\vect{q}\beta
\leq
\vect{s}_{k+1}^{0}
+(H-GL)\vect{d}_{k}^{\sigma},
\label{eq:signed_primal_con}\\
&\beta\leq\rho_k,
\label{eq:signed_primal_delta}
\end{align}
\end{subequations}
where $\beta$ is the signed reserve adjustment, $\vect{\eta}$ parameterizes the correction in $\range(T_m)$, and \eqref{eq:signed_primal_delta} guarantees reserve nonnegativity. When feasible, an optimizer $\vect{\eta}_{k,m}^{\sigma\star}$ defines
\begin{equation}
\underline{\rho}_{k+1,m}^{\sigma}
\coloneqq
\rho_k-\beta_{k,m}^{\sigma\star},
\quad
\vect{z}_{k+1,m}^{\sigma}
\coloneqq
\bar{\vect{z}}_{k+1}^{\sigma}
+T_m\vect{\eta}_{k,m}^{\sigma\star}.
\label{eq:primal_lower_def}
\end{equation}
Thus, the repaired plan $\vect{z}_{k+1,m}^{\sigma}$ combines the shifted plan $\vect{z}_{k+1}^{0}$, baseline update $L\vect{d}_{k}^{\sigma}$, and optimized correction $T_m\vect{\eta}_{k,m}^{\sigma\star}$. Here, $\beta_{k,m}^{\sigma\star}>0$ indicates the minimum required reserve reduction, whereas $\beta_{k,m}^{\sigma\star}<0$ certifies a reserve gain. If a reduced level is infeasible, we set $\underline\rho_{k+1,m}^{\sigma}=0$ as a fallback and expand the primal space.

\begin{proposition}[Signed primal hierarchy]
\label{prop:primal_hierarchy}
Let $\vect{x}_{k+1}^{\sigma}\in\cX_{\mathrm f}$ and $\sigma\in\{0,e\}$. For $m=0,\ldots,M_{\mathrm p}-1$, the primal hierarchy satisfies
\begin{subequations}
\label{eq:primal_prop}
\begin{align}
0 \leq \underline{\rho}_{k+1,m}^{\sigma} &\leq \rho(\vect{x}_{k+1}^{\sigma}), &\quad&\text{(Validity)} \label{eq:primal_validity}\\
\underline{\rho}_{k+1,m}^{\sigma} &\leq \underline{\rho}_{k+1,m+1}^{\sigma}, &\quad&\text{(Monotonicity)} \label{eq:primal_monotonicity}\\
\underline{\rho}_{k+1,M_{\mathrm p}}^{\sigma} &= \rho(\vect{x}_{k+1}^{\sigma}), &\quad&\text{(Finite exactness)} \label{eq:primal_exactness}
\end{align}
\end{subequations}
and whenever \eqref{eq:signed_primal_LP} is feasible, $(\vect{z}_{k+1,m}^{\sigma},\underline{\rho}_{k+1,m}^{\sigma})$ is primal-feasible for \eqref{eq:reserveLP} at $\vect{x}_{k+1}^{\sigma}$.
\end{proposition}

\begin{proof}
By \eqref{eq:slack} and \eqref{eq:successor_notation}, \eqref{eq:signed_primal_con} rewrites as $G\vect{z}_{k+1,m}^{\sigma}+\vect{q}(\rho_k-\beta)\leq\rfun(\vect{x}_{k+1}^{\sigma})$, with \eqref{eq:signed_primal_delta} ensuring $\rho_k-\beta\ge 0$. Thus, any optimizer yields a primal-feasible pair for \eqref{eq:reserveLP}, establishing \eqref{eq:primal_validity} (the fallback $\underline\rho=0$ is valid as $\rho\ge0$). By \eqref{eq:primal_nesting}, a feasible level remains feasible and $\beta^{\star}$ cannot increase under expansion; if level $m$ is infeasible, its fallback is zero. Hence \eqref{eq:primal_monotonicity} holds. Finally, for $T_{M_{\mathrm p}}=I_{n_z}$, setting $\vect{z}=\bar{\vect{z}}_{k+1}^{\sigma}+\vect{\eta}$ and $\delta=\rho_k-\beta$ recovers \eqref{eq:reserveLP}, proving \eqref{eq:primal_exactness}.
\end{proof}

\begin{corollary}[Cumulative reserve guarantee]
\label{cor:cumulative}
Let $(\vect{z}_{0}^{c},\hat{\rho}_0)$ be primal-feasible at $\vect{x}_0$. Suppose that at each time $k$ the current certified plan is shifted under some $\vect{w}_{k}^{0}\in\hat{\rho}_k\cW_0$ and a feasible signed primal level returns $\hat{\beta}_{k}^{\star}$. Then, with $\hat{\rho}_{k+1}\coloneqq\hat{\rho}_{k}-\hat{\beta}_{k}^{\star}$, the repaired plan remains primal-feasible at $\vect{x}_{k+1}$ and
\begin{equation}
\rho(\vect{x}_t)
\geq
\hat{\rho}_t
=
\hat{\rho}_0
-
\sum\nolimits_{k=0}^{t-1}\hat{\beta}_{k}^{\star},
\qquad \forall t\geq1.
\label{eq:cumulative}
\end{equation}
\end{corollary}

\begin{proof}
Repeated application of Proposition~\ref{prop:primal_hierarchy} proves the claim by induction.
\end{proof}
\subsection{Dual correction hierarchy}

\par To construct the dual hierarchy, use the current optimal multiplier $\vect{\lambda}_{k}^{\star}\in\cD$ as the dual anchor. Since $G^\top\vect{\lambda}_{k}^{\star}=\vect{0}$, restricting its correction to $\ker(G^\top)$ preserves the dual equality automatically. Let $n_\lambda\coloneqq\dim\ker(G^\top)$, and for $j=0,\ldots,M_{\mathrm d}$ let $V_j\in\R^{n_c\times p_j}$ be a basis matrix for the level-$j$ dual correction subspace, with $G^\top V_j=0$. Choose these spaces such that
\begin{equation}
\range(V_0)\subseteq\cdots\subseteq
\range(V_{M_{\mathrm d}})=\ker(G^\top).
\label{eq:dual_nesting}
\end{equation}
With $\range(V_0)\coloneqq\{\vect{0}\}$, level $j=0$ admits no correction, while increasing $j$ progressively releases additional equality-preserving directions until the full nullspace is recovered. As in the primal hierarchy, the choice of $V_j$ affects reduced-level performance, while the analysis below requires only the nesting and full-nullspace endpoint in
\eqref{eq:dual_nesting}.

\par Restricting the correction to $\range(V_j)$ at level $j$,
define the reduced dual LP
\begin{subequations}
\label{eq:reduced_dual_LP}
\begin{align}
\overline{\rho}_{k+1,j}^{\sigma}
\coloneqq
\min_{\vect{\xi}}&\quad
(\vect{\lambda}_{k}^{\star}+V_j\vect{\xi})^\top
\rfun(\vect{x}_{k+1}^{\sigma})
\label{eq:reduced_dual_obj}\\
\mathrm{s.t.}\quad&
\vect{\lambda}_{k}^{\star}+V_j\vect{\xi}\geq\vect{0},\quad
\vect{q}^\top(\vect{\lambda}_{k}^{\star}+V_j\vect{\xi})\geq1.
\label{eq:reduced_dual_cons}
\end{align}
\end{subequations}
Here, $\vect{\xi}\in\R^{p_j}$ parameterizes the equality-preserving correction in $\range(V_j)$; the constraints retain dual feasibility, while the objective tightens the upper reserve certificate. Since $\vect{\lambda}_{k}^{\star}\in\cD$, $\vect{\xi}=\vect{0}$ is feasible at every level, so the hierarchy always provides an upper certificate, and $j=0$ recovers the original recycled-dual bound $\overline{\rho}_{k+1,0}^{\sigma}
=
(\vect{\lambda}_{k}^{\star})^\top
\rfun(\vect{x}_{k+1}^{\sigma}).$

\begin{proposition}[Reduced dual hierarchy]
\label{prop:dual_hierarchy}
Let $\vect{x}_{k+1}^{\sigma}\in\cX_{\mathrm f}$ and $\sigma\in\{0,e\}$. For $j=0,\ldots,M_{\mathrm d}-1$, the dual hierarchy satisfies
\begin{subequations}
\label{eq:dual_prop}
\begin{align}
\rho(\vect{x}_{k+1}^{\sigma})
&\leq
\overline{\rho}_{k+1,j}^{\sigma},
&\quad&\text{(Validity)}
\label{eq:dual_validity}\\
\overline{\rho}_{k+1,j+1}^{\sigma}
&\leq
\overline{\rho}_{k+1,j}^{\sigma},
&\quad&\text{(Monotonicity)}
\label{eq:dual_monotonicity}\\
\overline{\rho}_{k+1,M_{\mathrm d}}^{\sigma}
&=
\rho(\vect{x}_{k+1}^{\sigma}),
&\quad&\text{(Finite exactness)}
\label{eq:dual_exactness}
\end{align}
\end{subequations}
\end{proposition}

\begin{proof}
Every multiplier feasible for \eqref{eq:reduced_dual_LP} belongs to $\cD$, so Lemma~\ref{lem:duality} gives \eqref{eq:dual_validity}. The nesting in \eqref{eq:dual_nesting} enlarges the reduced dual feasible family as $j$ increases, so the minimum cannot increase, proving \eqref{eq:dual_monotonicity}. At the terminal level, any $\vect{\lambda}\in\cD$ satisfies $\vect{\lambda}-\vect{\lambda}_{k}^{\star}\in\ker(G^\top)=\range(V_{M_{\mathrm d}})$. Thus \eqref{eq:reduced_dual_LP} recovers the full dual problem \eqref{eq:dual_problem}, proving \eqref{eq:dual_exactness}.
\end{proof}

\begin{remark}[Dual-space construction]
\label{rem:dual_construction}
Any nested family satisfying \eqref{eq:dual_nesting} is admissible. A deterministic hierarchy can be obtained by ordering a QR/SVD basis of $\ker(G^\top)$. To improve contraction at small $p_j$, representative offline multiplier-difference directions may instead be prioritized by observed upper-gap reduction and completed to the full nullspace.
\end{remark}

\par Under the default prefix construction, primal level $m=0,\ldots,N$ uses $mn_u+1$ decision variables, while dual level $j$ uses $p_j$ variables; the two hierarchies can be expanded independently. The terminal levels recover the full primal and dual reserve LPs and are not claimed to be cheaper than re-optimization; rather, they provide a guaranteed exact fallback, while computational savings arise when the desired certificate tolerance is reached at reduced levels.
\section{Anytime Reserve-Depletion Certification}

\par Using the statewise certificates above, define the signed reserve depletion induced by the excess perturbation as
\begin{equation}
\Delta_k
\coloneqq
\rho(\vect{x}_{k+1}^{0})
-
\rho(\vect{x}_{k+1}^{e}).
\label{eq:depletion}
\end{equation}
Positive $\Delta_k$ denotes reserve depletion, whereas negative $\Delta_k$ denotes reserve gain. Let $m_\sigma\in\{0,\ldots,M_{\mathrm p}\}$ and $j_\sigma\in\{0,\ldots,M_{\mathrm d}\}$ be independently chosen primal and dual levels for $\sigma\in\{0,e\}$. Define the online-computable statewise gap
\begin{equation}
\Gamma_{k+1}^{\sigma}(m_\sigma,j_\sigma)
\coloneqq
\overline{\rho}_{k+1,j_\sigma}^{\sigma}
-
\underline{\rho}_{k+1,m_\sigma}^{\sigma}.
\label{eq:state_gap}
\end{equation}

\begin{theorem}[Two-sided reserve-depletion envelope]
\label{thm:nested_certificate}
Suppose Assumption~\ref{ass:shift} holds and
$\vect{x}_{k+1}^{0},\vect{x}_{k+1}^{e}\in\cX_{\mathrm f}$.
Then, for each $\sigma\in\{0,e\}$,
\begin{equation}
\underline{\rho}_{k+1,m_\sigma}^{\sigma}
\leq
\rho(\vect{x}_{k+1}^{\sigma})
\leq
\overline{\rho}_{k+1,j_\sigma}^{\sigma}.
\label{eq:two_sided_state}
\end{equation}
Hence $\Gamma_{k+1}^{\sigma}(m_\sigma,j_\sigma)$ bounds either one-sided reserve error. Consequently,
\begin{equation}
\underline{\Delta}_k(\vect{m},\vect{j})
\leq
\Delta_k
\leq
\overline{\Delta}_k(\vect{m},\vect{j}),
\label{eq:depletion_envelope}
\end{equation}
where $\vect{m}=(m_0,m_e)$, $\vect{j}=(j_0,j_e)$, and
\begin{align}
\underline{\Delta}_k(\vect{m},\vect{j})
&\coloneqq
\underline{\rho}_{k+1,m_0}^{0}
-
\overline{\rho}_{k+1,j_e}^{e},
\label{eq:depletion_lower}\\
\overline{\Delta}_k(\vect{m},\vect{j})
&\coloneqq
\overline{\rho}_{k+1,j_0}^{0}
-
\underline{\rho}_{k+1,m_e}^{e}.
\label{eq:depletion_upper}
\end{align}
The envelope width satisfies
\begin{equation}
\begin{aligned}
W_k(\vect{m},\vect{j})
&\coloneqq
\overline{\Delta}_k-\underline{\Delta}_k\\
&=
\Gamma_{k+1}^{0}(m_0,j_0)
+
\Gamma_{k+1}^{e}(m_e,j_e),
\end{aligned}
\label{eq:width}
\end{equation}
and is nonincreasing in each of $m_0,m_e,j_0,j_e$. If, for each $\sigma\in\{0,e\}$, either $m_\sigma=M_{\mathrm p}$ or $j_\sigma=M_{\mathrm d}$, then the corresponding state reserve is exact. For scalar envelope evaluation, the companion bound may then be set to the same exact value, yielding
\begin{equation}
W_k=0,
\qquad
\underline{\Delta}_k
=
\overline{\Delta}_k
=
\Delta_k.
\label{eq:finite_exactness}
\end{equation}
Hence, any expansion sequence that advances a nonterminal index at each unresolved iteration remains valid after every solve and becomes exact in finitely many expansions.
\end{theorem}

\begin{proof}
Equation~\eqref{eq:two_sided_state} follows from Propositions~\ref{prop:primal_hierarchy} and \ref{prop:dual_hierarchy}. Combining the modeled-state lower bound with the realized-state upper bound gives \eqref{eq:depletion_lower}, while reversing the bound choices gives \eqref{eq:depletion_upper}. Direct subtraction yields \eqref{eq:width}. Primal expansion raises a lower bound and dual expansion lowers an upper bound, so $W_k$ is nonincreasing. Finally, \eqref{eq:primal_exactness} and \eqref{eq:dual_exactness} give the exact endpoint values and hence \eqref{eq:finite_exactness}.
\end{proof}

\par Theorem~\ref{thm:nested_certificate} remains valid across active-set changes. If the updated basic solution is primal-feasible, the corresponding successor reserve is exact without further hierarchy solves.

\begin{corollary}[Basis-preserving exactness]
\label{cor:basis}
Consider a standard form representation of \eqref{eq:reserveLP}, and let $\cB_k$ be an optimal basis at $\vect{x}_k$ with associated dual vector $\vect{\lambda}_{\cB_k}$. For $\sigma\in\{0,e\}$, let $(\vect{z}_{\cB_k}^{\sigma},\delta_{\cB_k}^{\sigma})$ be the basic solution obtained after replacing the right-hand side by $\rfun(\vect{x}_{k+1}^{\sigma})$. If this basic solution is primal-feasible, then
\begin{equation}
\delta_{\cB_k}^{\sigma}
=
(\vect{\lambda}_{\cB_k})^\top
\rfun(\vect{x}_{k+1}^{\sigma})
=
\rho(\vect{x}_{k+1}^{\sigma}).
\label{eq:basis_exact}
\end{equation}
Hence, the corresponding state gap is zero without solving a reduced primal or dual LP.
\end{corollary}

\begin{proof}
Since only the right-hand side changes, the basis remains dual-feasible; primal feasibility of the updated basic solution therefore implies optimality.
\end{proof}

\begin{algorithm}[t]
\caption{Basis-First Anytime Reserve Certification}
\label{alg:anytime}
\small
\begin{algorithmic}[1]
\Require successor data, $\vect{\lambda}_k^\star$,
$\{T_m\}_{0}^{M_{\mathrm p}}$, $\{V_j\}_{0}^{M_{\mathrm d}}$,
tolerance $\varepsilon_{\mathrm{cert}}\geq0$, optional basis $\cB_k$
\Ensure envelope $[\underline{\Delta}_k,\overline{\Delta}_k]$
and feasible excess-state repair $\vect{z}_{k+1}^{e}$

\State $(\mathcal C) \gets\textsc{InitializeBounds}(\cB_k)$
\Comment{basis exact or level-zero certificates}
\State form $[\underline{\Delta}_k,\overline{\Delta}_k]$ and $W_k$
from \eqref{eq:depletion_lower}--\eqref{eq:width}

\While{$\neg\textsc{AnytimeStop}(\mathcal C,W_k,\varepsilon_{\mathrm{cert}})$}
\State $(\sigma^\star,a^\star)\gets\textsc{SelectExpansion}(\mathcal C)$,
\quad $a^\star\in\{\mathrm p,\mathrm d\}$
\State increment $m_{\sigma^\star}$ if $a^\star=\mathrm p$;
otherwise increment $j_{\sigma^\star}$
\State solve the selected reduced LP
\eqref{eq:signed_primal_LP} or \eqref{eq:reduced_dual_LP}
and update its statewise certificate
\State update $\vect{z}_{k+1}^{e}$ in $\mathcal C$ if a feasible excess-state repair is obtained
\State apply \eqref{eq:primal_exactness} or \eqref{eq:dual_exactness}
if $m_{\sigma^\star}=M_{\mathrm p}$ or $j_{\sigma^\star}=M_{\mathrm d}$
\State update $\Gamma_{k+1}^{\sigma^\star}$,
$[\underline{\Delta}_k,\overline{\Delta}_k]$, and $W_k$
\EndWhile

\State \Return $[\underline{\Delta}_k,\overline{\Delta}_k]$ and $\vect{z}_{k+1}^{e}$
\end{algorithmic}
\end{algorithm}

\par Theorem~\ref{thm:nested_certificate} leads to the basis-first anytime implementation in Algorithm~\ref{alg:anytime}, which balances two objectives: tightening the depletion envelope $W_k$ to the prescribed tolerance $\varepsilon_{\mathrm{cert}}$, and retaining an implementable primal repair $\vect{z}_{k+1}^{e}$ for the realized successor $\vect{x}_{k+1}^{e}$. Here, $\vect{z}_{k+1}^{e}$ denotes the latest feasible excess-state repair obtained either from a successful basis test or the $m_e$-level excess-state primal solve. Let $\mathcal C$ collect the current hierarchy levels, statewise bounds, exactness status, and the latest feasible repair $\vect{z}_{k+1}^{e}$. For each successor $\sigma\in\{0,e\}$, $\textsc{InitializeBounds}$ first applies the basis test in Corollary~\ref{cor:basis}. If the updated basic solution is primal-feasible, it yields the exact reserve without solving a hierarchy LP (and $\vect{z}_{k+1}^{e}$ is cached when $\sigma=e$); otherwise, the level-zero primal and recycled-dual certificates are computed. When hierarchy expansion is needed, $\textsc{SelectExpansion}$ outputs $(\sigma^\star,a^\star)$, where $a^\star\in\{\mathrm p,\mathrm d\}$ denotes primal or dual expansion. If a required repair $\vect{z}_{k+1}^{e}$ is unavailable, it prioritizes the excess-state primal hierarchy ($\sigma^\star=e, a^\star=\mathrm p$); thereafter, unresolved levels can be selected via an implementation-specific rule (e.g., by selecting the state with the larger gap $\Gamma_{k+1}^{\sigma}$) to tighten $W_k$. While this selection policy governs contraction speed, it leaves theoretical validity and finite exactness unaffected. The algorithm terminates via $\textsc{AnytimeStop}$ when $W_k\leq\varepsilon_{\mathrm{cert}}$ and a feasible repair $\vect{z}_{k+1}^{e}$ is secured—a necessary condition since exact dual bounds certify reserve values but do not yield a primal-feasible input for closed-loop execution. Consequently, Algorithm~\ref{alg:anytime} operates in an anytime manner: it maintains a valid depletion envelope after every completed hierarchy solve, supports early termination once the prescribed tolerance is met, and reaches the exact depletion in finitely many levels if it continues.
\section{Numerical Study}
We evaluate reserve geometry, finite-perturbation validity, and the computation--tightness tradeoff of Algorithm~\ref{alg:anytime}. Exact successor reserves are obtained by full re-optimization only for assessment and are not used by the online certificate.
\subsection{Double integrator: geometry and finite perturbations}
Consider a double integrator
\[
A=\begin{bmatrix}1&T_s\\0&1\end{bmatrix},\,
B=E=\begin{bmatrix}T_s^2/2\\T_s\end{bmatrix},\,
\cW_0=[-1,1].
\]
We use $T_s=0.1$~s, $N=15$, $|p|\leq4$~m, $|v|\leq2$~m/s, and $|u|\leq1$~m/s$^2$. The fixed ancillary feedback and the prediction map in \eqref{eq:Lcl} use the same discrete LQR gain $K_r=[-2.738\ -2.641]$; support-function tightenings and robust invariant terminal ingredients constructed under $A+BK_r$ ensure Assumption~\ref{ass:shift}. Figure~\ref{fig:double_integrator}(a) shows the exact reserve on a $151\times151$ position--velocity grid. The two equal-position-margin states have $\rho_A=8.90\times10^{-4}$ and $\rho_B=3.19\times10^{-2}$, differing by a factor of $35.8$.

Finite perturbations are generated as $w_k=\alpha\rho_k s$, with $s\in\{-1,1\}$ and $\alpha$ stratified over $[1.02,2]$, retaining only pairs satisfying the domain conditions of Theorem~\ref{thm:nested_certificate}. In the recycled-basis-failure case of Fig.~\ref{fig:double_integrator}(b), the local sensitivity prediction $\widehat\Delta_k^{\mathrm{sens}}:=-(\vect{\lambda}_k^\star)^\top H\vect{d}_k$ incurs an absolute error of $2.29\times10^{-3}$. The envelope contains the exact $\Delta_k$ after every completed solve and reaches the prescribed tolerance after $6$ reduced solves, before either exact endpoint. Over the recycled-basis-failure subset, the median normalized sensitivity error is $1.16\times10^{-3}$. Across $4000$ pairs, no bound-validity, monotonicity, endpoint-exactness, or repair-feasibility violation exceeded $10^{-7}$; And approximately $60.3\%$ of the feasible signed repairs satisfied $\beta_{k,m}^{\sigma\star}<0$.
\begin{figure}[H]
\centering
\includegraphics[width=\columnwidth]{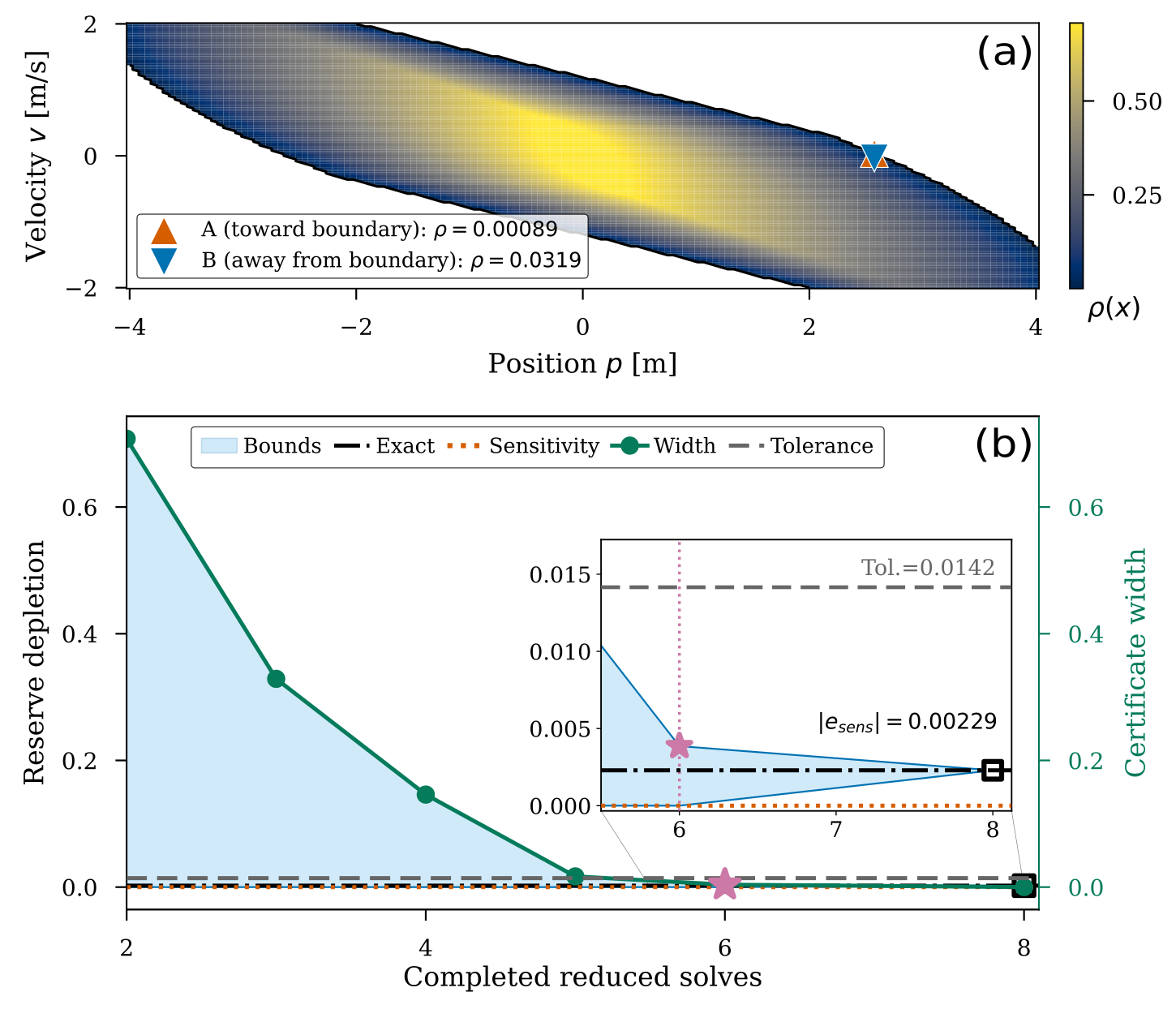}
\caption{Double-integrator results. (a) Equal geometric margins do not imply equal certified reserves. (b) Following a recycled-basis failure, the certified envelope contains the exact depletion and reaches the prescribed tolerance before either exact endpoint, whereas the sensitivity prediction is not exact.}
\label{fig:double_integrator}
\end{figure}
\subsection{Lateral-vehicle benchmark}
The medium-scale benchmark is a zero-order-hold discretized linear bicycle model with $\vect{x}=[e_y,e_\psi,v_y,r]^\top$, $u=\delta_f$, and $\vect{w}=[w_F,w_M]^\top\in\delta[-1,1]^2$. We use $V_x=15$~m/s, $T_s=0.05$~s, $N\in\{20,40,60\}$, and $(m,I_z,a,b,C_f,C_r)=(1575,2875,1.2,1.6,8\times10^4,8\times10^4)$ in SI units; $w_F$ and $w_M$ represent a $500$~N lateral force and a $300$~N\,m yaw moment. The constraints are $|e_y|\leq1$~m, $|e_\psi|\leq0.12$~rad, $|v_y|\leq1.5$~m/s, $|r|\leq0.4$~rad/s, and $|\delta_f|\leq0.18$~rad. The LQR gain is $K_r=[-0.5825\ -2.7888\ -0.0822\ -0.1529]$.

The primal prefix length and the number of retained directions in a deterministic QR basis of $\ker(G^\top)$ are increased geometrically, with terminal levels included explicitly. \footnote{Specifically, primal and dual dimensions follow $\{0,1,2,4,\dots,N\}$ and $\{0,4,8,16,\dots,n_\lambda=\dim\ker(G^\top)\}$, respectively. If the $N$-prefix is not full-dimensional in $\R^{n_z}$, $T_{M_{\mathrm p}}=I_{n_z}$ is appended.} If no excess-state repair exists, \textsc{SelectExpansion} expands $m_e$; otherwise, it selects the larger statewise gap and expands the side adding fewer variables, alternating ties. The stopping tolerance and normalized certificate width we used are $\varepsilon_{\mathrm{cert}}=\max\{10^{-5},0.02\rho_k\}, \widetilde W_k=\frac{W_k}{\max\{\rho_k,10^{-6}\}}$. For each horizon, the held-out set contains $3000$ natural pairs and $1000$ additional pairs for which at least one recycled-basis test fails; all retained successors lie in $\cX_{\mathrm f}$. We compare the level-zero certificate $\mathrm R_0$, primal-only $\mathrm P_{\varepsilon}$, control-compatible dual-only $\mathrm D_{\varepsilon}$, $\mathrm{PD}_{\varepsilon}^{-\mathrm B}$ without the basis test, the complete proposed $\mathrm{PD}_{\varepsilon}$ method, and two warm-started full successor solves. All methods use the same single-threaded HiGHS dual-simplex backend and $10^{-8}$ feasibility tolerances; timing includes projection, basis tests, right-hand-side updates, selection, and all solver calls, but excludes offline construction.

\begin{figure}[!t]
\centering
\includegraphics[width=\columnwidth]{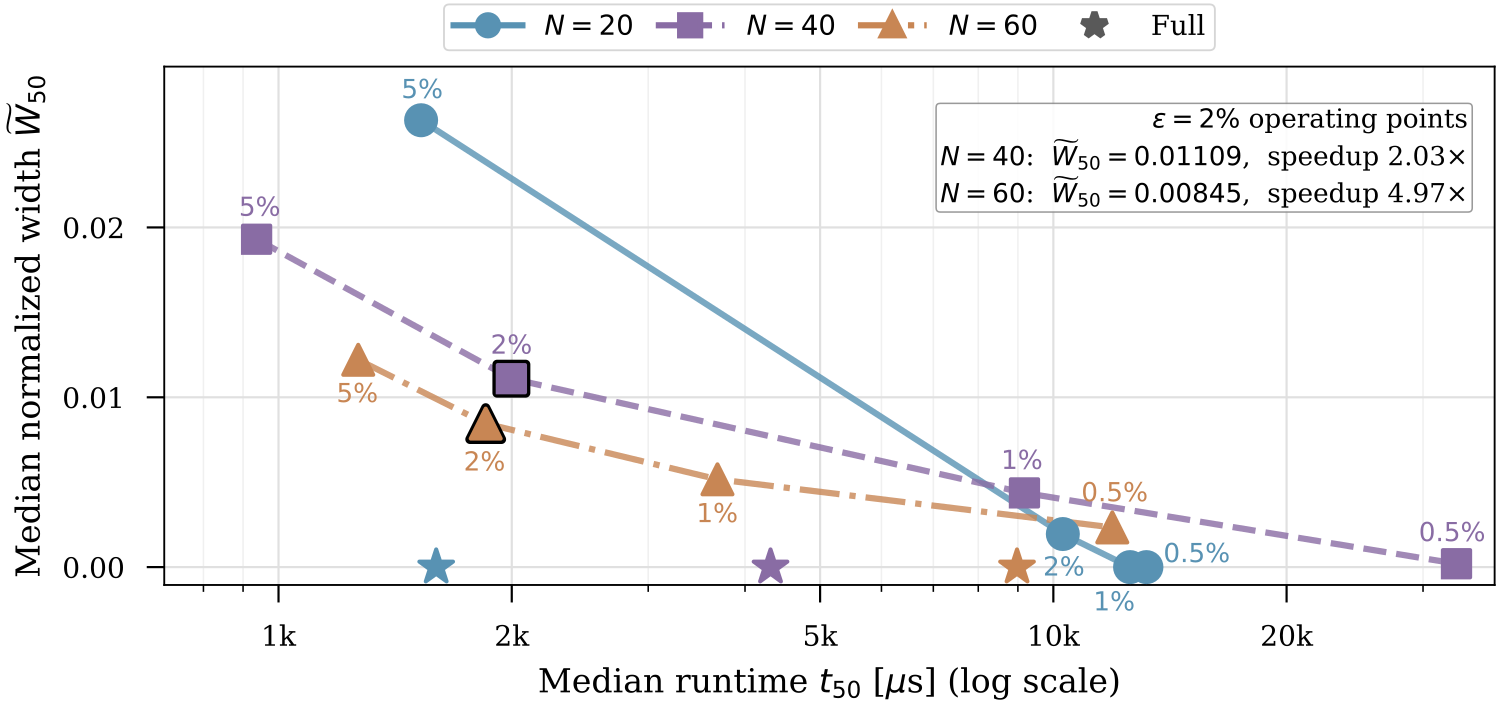}
\caption{Vehicle certificate tightness--computation tradeoff for $N\in\{20,40,60\}$; stars denote warm-started \textit{Full} as the zero-width endpoint. At $\varepsilon_{\mathrm{cert}}=2\%$, the highlighted $N=40$ and $N=60$ operating points attain median normalized widths of $0.01109$ and $0.00845$, with paired median speedups $\operatorname{median}_i(t_{\mathrm{Full},i}/t_{\mathrm{PD},i})$ of $2.03\times$ and $4.97\times$, respectively.}
\label{fig:vehicle_pareto}
\end{figure}

Figure~\ref{fig:vehicle_pareto} summarizes the computation--tightness tradeoff across horizons. At $\varepsilon_{\mathrm{cert}}=2\%$, $N=40$ and $N=60$ attain median normalized widths of $0.01109$ and $0.00845$, with paired median speedups of $2.03\times$ and $4.97\times$, respectively, showing that useful certified widths can be reached before the zero-width \textit{Full} endpoint and that the benefit grows with horizon. Table~\ref{tab:vehicle} reports the $N=40$ results. All methods have zero certificate violations; $\mathrm{PD}_{\varepsilon}$ stops before either exact endpoint in $92.3\%$ of cases and has the lowest $t_{50}$ and $t_{95}$ and the highest early-stop rate among the hierarchy-based methods. The ablations are complementary: $\mathrm P_{\varepsilon}$ is faster but looser, whereas $\mathrm D_{\varepsilon}$ is tighter but costlier under the present QR ordering, so their adaptive combination gives the best time-to-tolerance tradeoff. Warm-started \textit{Full} retains the lower $t_{95}$ because hard cases may require several reduced expansions. Problem-informed dual directions and a cost-aware switch to \textit{Full} can reduce this tail without changing validity, monotonicity, or finite exactness. The basis test gives exact reserves for $11.6\%$ of the modeled and $10.2\%$ of the realized successors. Finally, across $100$ domain-conditioned $40$-step closed-loop runs, only $5$ of $4005$ disturbance candidates were rejected and no trajectory restart was required. The cumulative bound in \eqref{eq:cumulative} held at all $4000$ accepted states within $10^{-7}$, with a maximum numerical shortfall $1.83\times10^{-13}$; moreover, $38.1\%$ of repairs satisfied $\hat{\beta}_k^\star<0$, demonstrating the capture of both reserve depletion and gain.

\begin{table}[!t]
\caption{End-to-end vehicle certificate performance}
\label{tab:vehicle}
\centering
\setlength{\tabcolsep}{7pt}
\begin{threeparttable}
\begin{tabular}{@{}lccccc@{}}
\toprule
Method & Viol. & $\widetilde W_{50}$ & $t_{50}$ [$\mu$s] & $t_{95}$ [$\mu$s] & Early [\%]\\
\midrule
$\mathrm R_0$ & 0 & 0.0321 & 878.3 & 927.5 & -- \\
$\mathrm P_{\varepsilon}$ & 0 & 0.0139 & 2219.8 & 15911.7 & 88.1 \\
$\mathrm D_{\varepsilon}$ & 0 & 0.0099 & 9153.8 & 23593.3 & 75.6 \\
$\mathrm{PD}_{\varepsilon}^{-\mathrm B}$ & 0 & 0.0117 & 2273.0 & 14200.3 & 89.9 \\
$\mathrm{PD}_{\varepsilon} (\text{our})$ & 0 & 0.0111 & 2003.5 & 10500.4 & 92.3 \\
Full & 0 & 0 & 4315.7 & 5435.8 & 0 \\
\bottomrule
\end{tabular}
    \begin{tablenotes}
      \item \scriptsize Evaluated with $N=40$, $\varepsilon_{\mathrm{cert}}=2\%$, and violation threshold $10^{-7}$. Early: stopping before either exact endpoint.
    \end{tablenotes}
\end{threeparttable}
 \vspace{-15pt}
\end{table}
\section{Conclusion}
This letter recasts post-disturbance reserve assessment as an interruptible certification task, replacing immediate full re-optimization with a tunable depletion envelope, an implementable successor repair, and an exact fallback. It is well suited to real-time MPC with limited or variable computation, where a certified nonzero-width estimate is sufficient; the experiments identify moderate tolerances and longer horizons as a favorable operating regime. Future work will develop problem-informed dual hierarchies and cost-aware switching to warm-started full re-optimization to improve gap reduction and tail latency. Extensions to nonlinear, time-varying, and adaptive uncertainty-set MPC are also of interest.
 \vspace{-8pt}
\bibliographystyle{IEEEtran}
\bibliography{References}

@article{mayne2005robust,
  title={Robust model predictive control of constrained linear systems with bounded disturbances},
  author={Mayne, David Q and Seron, Mar{\'\i}a M and Rakovi{\'c}, Sa{\v{s}}a V},
  journal={Automatica},
  volume={41},
  number={2},
  pages={219--224},
  year={2005},
  publisher={Elsevier}
}

@article{zhang2017robust,
  title={Robust optimal control with adjustable uncertainty sets},
  author={Zhang, Xiaojing and Kamgarpour, Maryam and Georghiou, Angelos and Goulart, Paul and Lygeros, John},
  journal={Automatica},
  volume={75},
  pages={249--259},
  year={2017},
  publisher={Elsevier}
}

@inproceedings{kim2018robust,
  title={Robust model predictive control with adjustable uncertainty sets},
  author={Kim, Yeojun and Zhang, Xiaojing and Guanetti, Jacopo and Borrelli, Francesco},
  booktitle={2018 IEEE Conference on Decision and Control (CDC)},
  pages={5176--5181},
  year={2018},
  organization={IEEE}
}

@inproceedings{raghuraman2021tube,
  title={Tube-based robust MPC with adjustable uncertainty sets using zonotopes},
  author={Raghuraman, Vignesh and Koeln, Justin P},
  booktitle={2021 American Control Conference (ACC)},
  pages={462--469},
  year={2021},
  organization={IEEE}
}

@inproceedings{kerrigan2001robust,
  title={Robust feasibility in model predictive control: Necessary and sufficient conditions},
  author={Kerrigan, Eric C and Maciejowski, Jan M},
  booktitle={2001 IEEE 40th Conference on Decision and Control (CDC)},
  pages={728--733},
  year={2001},
  organization={IEEE}
}

@article{parsi2024once,
  title={Once upon a time step: A closed-loop approach to robust MPC design},
  author={Parsi, Anilkumar and Bartos, Marcell and Srivastava, Amber and Gros, Sebastien and Smith, Roy S},
  journal={IEEE Transactions on Automatic Control},
  volume={70},
  number={2},
  pages={1297--1303},
  year={2024},
  publisher={IEEE}
}

@inproceedings{zhang2019safe,
  title={Safe and near-optimal policy learning for model predictive control using primal-dual neural networks},
  author={Zhang, Xiaojing and Bujarbaruah, Monimoy and Borrelli, Francesco},
  booktitle={2019 American Control Conference (ACC)},
  pages={354--359},
  year={2019},
  organization={IEEE}
}

@inproceedings{skibik2021feasibility,
  title={Feasibility governor for linear model predictive control},
  author={Skibik, Terrence and Liao-McPherson, Dominic and Cunis, Torbj{\o}rn and Kolmanovsky, Ilya and Nicotra, Marco M},
  booktitle={2021 American Control Conference (ACC)},
  pages={2329--2335},
  year={2021},
  organization={IEEE}
}

@article{carson2013safety,
  title={A robust model predictive control algorithm augmented with a reactive safety mode},
  author={Carson III, John M and A{\c{c}}{\i}kme{\c{s}}e, Beh{\c{c}}et and Murray, Richard M and MacMartin, Douglas G},
  journal={Automatica},
  volume={49},
  number={5},
  pages={1251--1260},
  year={2013},
  publisher={Elsevier}
}

@article{tondel2003algorithm,
  title={An algorithm for multi-parametric quadratic programming and explicit MPC solutions},
  author={T{\o}ndel, Petter and Johansen, Tor Arne and Bemporad, Alberto},
  journal={Automatica},
  volume={39},
  number={3},
  pages={489--497},
  year={2003},
  publisher={Elsevier}
}

@article{ferreau2008online,
  title={An online active set strategy to overcome the limitations of explicit MPC},
  author={Ferreau, Hans Joachim and Bock, Hans Georg and Diehl, Moritz},
  journal={International Journal of Robust and Nonlinear Control},
  volume={18},
  number={8},
  pages={816--830},
  year={2008},
  publisher={Wiley Online Library}
}

@article{zeilinger2011realtime,
  title={Real-time suboptimal model predictive control using a combination of explicit MPC and online optimization},
  author={Zeilinger, Melanie Nicole and Jones, Colin Neil and Morari, Manfred},
  journal={IEEE transactions on automatic control},
  volume={56},
  number={7},
  pages={1524--1534},
  year={2011},
  publisher={IEEE}
}

@article{liao2020time,
  title={Time-distributed optimization for real-time model predictive control: Stability, robustness, and constraint satisfaction},
  author={Liao-McPherson, Dominic and Nicotra, Marco M and Kolmanovsky, Ilya},
  journal={Automatica},
  volume={117},
  pages={108973},
  year={2020},
  publisher={Elsevier}
}

@article{garone2017reference,
  title={Reference and command governors for systems with constraints: A survey on theory and applications},
  author={Garone, Emanuele and Di Cairano, Stefano and Kolmanovsky, Ilya},
  journal={Automatica},
  volume={75},
  pages={306--328},
  year={2017},
  publisher={Elsevier}
}

@article{diehl2005real,
  title={A real-time iteration scheme for nonlinear optimization in optimal feedback control},
  author={Diehl, Moritz and Bock, Hans Georg and Schl{\"o}der, Johannes P},
  journal={SIAM Journal on control and optimization},
  volume={43},
  number={5},
  pages={1714--1736},
  year={2005},
  publisher={SIAM}
}

@article{greenberg2000sensitivity,
  title={Simultaneous primal-dual right-hand-side sensitivity analysis from a strictly complementary solution of a linear program},
  author={Greenberg, Harvey J},
  journal={SIAM Journal on Optimization},
  volume={10},
  number={2},
  pages={427--442},
  year={2000},
  publisher={SIAM}
}

@article{hosseinzadeh2023robust,
  title={Robust-to-early termination model predictive control},
  author={Hosseinzadeh, Mehdi and Sinopoli, Bruno and Kolmanovsky, Ilya and Baruah, Sanjoy},
  journal={IEEE transactions on automatic control},
  volume={69},
  number={4},
  pages={2507--2513},
  year={2023},
  publisher={IEEE}
}

@inproceedings{pavlov2019early,
  title={Early termination of NMPC interior point solvers: Relating the duality gap to stability},
  author={Pavlov, Andrei and Shames, Iman and Manzie, Chris},
  booktitle={2019 18th European Control Conference (ECC)},
  pages={805--810},
  year={2019},
  organization={IEEE}
}

@article{xie2023maximal,
  title={Maximal admissible disturbance constraint set for tube-based model predictive control},
  author={Xie, Huahui and Dai, Li and Sun, Zhongqi and Xia, Yuanqing},
  journal={IEEE Transactions on Automatic Control},
  volume={68},
  number={11},
  pages={6773--6780},
  year={2023},
  publisher={IEEE}
}

\end{document}